\documentclass[a4paper]{article} 
\usepackage[utf8]{inputenc}  
\usepackage{amsmath}
\usepackage{amssymb}
\usepackage{amsthm}
\usepackage{color}
\newtheorem{theorem}{Theorem}
\newtheorem{lemma}{Lemma}

\newtheorem{corollary}{Corollary}
\usepackage{cite}
\usepackage{arydshln}

\date{}

\def\bea{\begin{eqnarray}}
\def\eea{\end{eqnarray}}
\def\be{\begin{equation}}
\def\ee{\end{equation}}
\def\bes{\begin{equation*}}
\def\ees{\end{equation*}}
\def\bs{\begin{split}}
\def\es{\end{split}}
\def\beas{\begin{eqnarray*}}
\def\eeas{\end{eqnarray*}}
\def\cse{{\mathcal S}}
\def\cT{{\mathcal T}}

\begin{document}
\title{Admissible Equivalence Transformations for Linearization of Nonlinear Wave Type Equations }
\author{Saadet S. Özer }

\maketitle {}
Istanbul Technical University, Faculty of Science and Letters,\\
Department of Mathematics, 34469 Maslak Istanbul
Turkey\\
email: saadet.ozer@itu.edu.tr\\ 
\begin {abstract}
In the present paper we consider a general family of two dimensional wave equations which represents a great variety of linear and nonlinear  equations within the framework of the transformations of  equivalence groups.  We have investigated the existence problem of  point transformations that lead mappings between linear and nonlinear members of particular families and determined the structure of the nonlinear terms of linearizable equations.  We have also given  examples about some equivalence transformations between linear and nonlinear equations and obtained  exact solutions of  nonlinear equations via the linear ones.

\end{abstract}
\vskip 5mm
{\bf keywords:}
Lie Group Application;  Equivalence Transformations; Exact Solution;  Nonlinear Wave Equation
\section{Introduction}
Several system of partial differential equations contain some
arbitrary functions or parameters. Such systems represent a family
of equations of the same structure. Almost all field equations in
continuum physics involve some arbitrary parameters since they
describe certain common or similar behaviors of diverse materials.
 Equivalence groups are groups of
transformations which leave a given family of equations 
invariant. They map an arbitrary member of the family onto another
member of the same family of equations, so that they transform a known solution
of an equation  into a solution of another equation in the same family. 
\\
It was first noted by Lie ~\cite{lie1} that every invariant system
of differential equations ~\cite{lie2} and every variational problem
~\cite{lie3} could be directly expressed in terms of differential
invariants. Lie's preliminary results were generalized by Tresse
~\cite{tresse} and the first systematic treatment was formulated by
Ovsiannikov ~\cite{ovs} who had observed that the usual Lie's
infinitesimal invariance approach could be employed in order to
construct equivalence groups. \\
Various well-developed methods have been used to construct equivalence transformations. The general theory of determining
transformation groups and algorithms  can be found in the references
\cite{ovs, olver, suhubiexterior, ibr1}. In the present work, we
employ the results obtained
 by \c Suhubi \cite{suhubisecond} for second order quasilinear partial differential equations.
 The method he used, depends on a
geometrical approach based on Cartan's \cite{car} formulation of
partial differential equations in terms of exterior differential
forms. It was first formulated in a seminal paper by
Harrison and Estabrook \cite{har}  and explored further by Edelen \cite{ede}. We can
refer  the reader for more general results and applications to
\cite{oze, suh2} and \cite{suh3}. The reader may also have a look at  the work on determining exact solutions of nonlinear diffusion equation by Özer \cite{ozer2018} to see some particular  applications.  \\
Nonlinear wave equations have been taken great interest by many researchers in the works of application of Lie Groups.  Equivalence groups of some members are examined to inversitage the group classification,  differential invariants, linearization problem and so on. Equivalence transformations for a special one dimensional nonlinear wave equation $u_{tt}-u_{xx}=f(u,u_t,u_x)$ are studied by Tracina \cite{tracina}, she has determined solutions of nonlinear equations via differential invariants. Group classification of a one dimensional family of equation with nonlinearity on the coefficient $u_{xx}$ is given by Song et.al. \cite{song}, in their study they have not investigate the  linearization problem. Ibragimov et al. \cite{ibrwave} have considered the  wave equation with nonlinear force term within the frame work of equivalence transformations, and a part of their result on properties of the equation to be quasi self-adjoint is supported by our results in this paper. To understand the algebraic structure of such problems, we reommend the reader the work of  Bihlo et al. \cite{bihlo}, in which a family of one dimesional nonlinear wave equation has been examined in detail.   Another study on the complete group classification of one dimensional nonlinear Klein-Gordon equations has been done by Long et al. \cite{longmeleshko}. Just recently, Ibragimov \cite{ibragimov2018} has considered the problem of determining exact solutions for a class of nonlinear two dimensional wave type equation by using conservation laws. 
 \\
 The aim of the present paper is to investigate the admissible point transformations
for a general family of two dimensional  wave
equations between its linear and nonlinear members.  We have obtained the infinitesimal generators of  equivalence groups for the most general family and then by applying  restrictions on the  nonlinear terms, we have determined the structures of  admissible transformations.  After giving some details about the method in the Introduction, in the second section we have obtained the determining equations for equivalence groups for the family. In the third section we have investigated the  functional dependencies
 of nonlinear equations that can be mapped onto linear equations via appropriate infinitesimal generators of the transformations.  The results are given in a table. In the last section we have  given some examples on particular maps and  details about how to find the exact solutions of nonlinear equations  via equivalence transformations. Since such wide class of  (2+1) dimensional  wave  equations has not been considered within the frame work of the Lie`s transformation groups in the literature before, the results obtained here are original and represent a general view of linearizable equations. 
 \\
A brief summary of this work was presented at the International Conference on Mathematics "An Istanbul Meeting for World Mathematicians" Istanbul, Turkey, 3-6 July 2018 \cite{ozerproceeding}.

\par

\section{Equivalence Transformations}
In the present paper we shall examine the   group of equivalence transformations of a general family of two dimensional  wave equations
\be
\label{main}
u_{tt}=f(x,y,t,u,u_x,u_y,u_t)_x + g(x,y,t,u,u_x,u_y,u_t)_y + h(x,y,t,u,u_x,u_y,u_t) 
\ee
which represents a great variety  of linear and nonlinear wave type equations. The study will focus on investigation of admissible transformations between linear and nonlinear wave equations that can be expressed as  members of the family of equations given by the equation \eqref{main}. Here $u$ is the dependent variable of the independent variables $x,y, t$ and $f, g$, $h$ are  smooth nonconstant functions of their variables, subscripts
denote the partial derivatives with respect to the corresponding variables. We shall investigate  all admissible point transformations $\bar x, \bar y, \bar t, \bar u, \bar f, \bar g, \bar h$ yield
\begin{equation}
\label{equivalence1}
 u_{tt}=f_x+g_y+h  \longrightarrow \bar u_{\bar t \bar t}=\bar f_{\bar x} +\bar g_{\bar y}+ \bar h
\end{equation}
where $\bar {(.)}$ denotes the transformed variables and functions. 
\\
Let $\mathcal{M}=\mathcal{N}\times \mathbb R$  be a (2+1) dimensional manifold with a local coordinate system $\mathbf{x}=( x_{i}) =( x,y,t)$ which we  call as the
space of independent variables. Consider a trivial bundle structure $(
\mathcal{K},\pi ,\mathcal{M}) $ with fibers are the real line $\mathbb{R}$. Here $\mathcal{M}$ is the base manifold and $\mathcal{K}$, called the
graph space is globally in form of a product manifold $\mathcal{M}\times \mathbb{R}$. We equip the four dimensional graph space $\mathcal{K}$ with
the local coordinates $( \mathbf{x},u) =( x,y,t,u) $.
\\
 A vector field on  $\mathcal{K}$ is a section of its tangent
bundle and locally  in form
 \be
V=\xi^1\frac{\partial}{\partial x}+ \xi^2\frac{\partial}{\partial y}+\xi^3\frac{\partial}{\partial t}+\eta\frac{\partial}{\partial u}
\ee
 where $\xi^1,\ \xi^2, \ \xi^3$ and $\eta$ are
 $$\xi^i=\xi^i(x,y,t,u),\qquad \eta=\eta(x,y,t,u).$$
 In order to construct the equivalence groups which are much more general than the classical  symmetry groups, first we  extend the graph space $\mathcal{K}$  by adding the auxiliary variables:
 \[
 \{v_j=\frac{\partial u}{\partial {x^j}}, \ f, \ g, \ h\}.
 \] And to recognize the   functional dependencies of the  functions $f,\ g$ and $h$ we also add
  \be
 \label{additional}
 \begin{split}
& s^1_{j}=f_{x^j},\ s^2 _{j}=g_{x^j},\ \sigma^1=f_u,\ \sigma^2=g_u,\  s^{1j}=f_{v_j}, \ s^{2j}=g_{v_j}, \\ &\qquad t_j= h_{x^j},\ t^j=h_{v_j},\ \tau=h_u  
\end{split}
\ee to $\mathcal{K}$, where $j=1,2,3$ and represent $x,y,t$, respectively. The coordinate cover of the extended manifold can now be written as
 \be
 \label{manifold}
 \tilde{\mathcal{K}}=\{x^j,u,f,g,h,v_j,s^1_j,s^2_j, \sigma^1,\sigma^2,  s^{1j}, s^{2j}, t_j, t^j, \tau\}.
 \ee
The prolongation vector $\tilde V$ over the extended manifold covered by $\mathcal{\tilde K}$   is written as
\begin{align}
 \label{prolonged}
 \tilde{V} &=V+\mu^1 \frac{\partial }{{\partial f }} +\mu^2 \frac{\partial }{{\partial g }}+ \mathcal{H} \frac{\partial}{\partial h} + \sum_{j=1}^3 \left[ \zeta_j
   \frac{\partial }{{\partial v_j }} +
     \sum_{i=1}^2 (S^i_j \frac{\partial }
   {{\partial s^i_j }} + S^{ij} \frac{\partial }
   {{\partial s^{ij} }})\right.\nonumber \\
   &\left. + T^j \frac{\partial}{\partial t^j} +  T_j \frac{\partial}{\partial t_j}\right]+ \sum_{i=1}^2 \mathcal{S}^i \frac{\partial }
   {{\partial \sigma^i }} +\mathcal{T}  \frac{\partial }
   {{\partial \tau }}
\end{align}
 where the coefficients  are  called the infinitesimal generators of the transformation group.  The equivalence group generated by the  vector field \eqref{prolonged} can then be determined by
integrating the following system of ordinary differential equations in terms of the group
parameter $\epsilon$
\be\label{system}
\begin{split}
&\frac{d\bar x^i}{d\epsilon} =\xi^i(\bar x^j, \bar u), \quad \frac{{d\bar u}}{{d\epsilon }} = \eta(\bar x^j, \bar u), \quad \frac{{d\bar f }}{{d\epsilon }} = \mu^1 (\bar x^j,\bar u,\bar v_j,  \bar f, \bar g, \bar h), \\
& \frac{{d\bar g }}{{d\epsilon }} = \mu^2 (\bar x^j,\bar u,\bar v_j,  \bar f, \bar g, \bar h),\quad \frac{{d\bar h }}{{d\epsilon }} = \mathcal{H} (\bar x^j,\bar u,\bar v_j,  \bar f, \bar g, \bar h)
\end{split}
\ee
under the initial conditions
 \be
 \label{initial}
 \bar x^j (0) = x^j ,\quad  \bar u(0) = u,\quad \bar f (0) = f
,\quad\bar g (0) = g, \quad \bar h(0)=h.
 \ee 
To determine the infinitesimal generators of the admissible transformation group for the general wave equation given by \eqref{main}, namely the coefficients of the prolonged vector field  \eqref{prolonged}, we shall use the method derived by Şuhubi \cite{suhubisecond} in the following section.  

\subsection{Determining Equations}
Note that the (2+1) dimensional general  wave equation \eqref{main} can be written as a first order equation as
\be
\label{firstequation}
{v_3}_{t}=f_x+g_y+h
\ee
which can be expressed as a single balance equation given as follows
\be
\label{balance}
\frac{\partial \Sigma^i}{\partial x^i}+\Sigma=0
\ee
by calling 
\begin{equation}
\label{notation}
 \Sigma^1= f, \quad \Sigma^2=g, \quad \Sigma^3=-v_3, \quad \Sigma=h.
\end{equation}
A vector field on the tangent space of the extended manifold  for the balance equation \eqref{balance} is written as follows:
\begin{equation}
 \label{prolongedbalance}
 \tilde{\tilde V} = \sum_{j=1}^3 \left(\xi^j \frac{\partial}{\partial x^j}+ \mu^j \frac{\partial }{{\partial \Sigma^j }} +   \zeta_j
   \frac{\partial }{{\partial v_j }} \right) +\eta \frac{\partial}{\partial u}+ \mathcal{H} \frac{\partial}{\partial \Sigma}
+\dots    
\end{equation}
where $\dots$ represents the terms related to the additional variables written by the functional dependence of $\Sigma^i$ and $\Sigma$, exactly in the same way in \eqref{additional}.
The generators formula for the coefficients of the vector field \eqref{prolongedbalance} related to the independent, dependent variables and the arbitrary functions in the family of equations  are determined in \cite{suhubisecond}   as
 \be
 \label{generators1}
\begin{split}
& \xi^i  = -\phi^i (x^j,u), \qquad \quad  \eta = \eta (x^j,u) ,\\
&\zeta_i=D_{x^i} \eta+(D_{x^i} \phi^j)v_j, \\
&
\mu^i = (w+\frac{{\partial \phi^j}}{{\partial u}}v_j)\Sigma^i-(D_j
\phi^i) \Sigma^j+ \alpha^{ij}v_j+\beta^i,\\
& \mathcal{H} = (w+\frac{\partial \phi^i}{\partial u}v_i)\Sigma -D_{x^i} \mu^i
 \end{split}
\ee where $D_{x^i}=\frac{\partial}{\partial
x^i}+v_i\frac{\partial}{\partial u}$, $\   \alpha^{ij}=-\alpha^{ji}$, and  $w, \ \alpha^{ij},\ \beta^i$ are  general smooth functions of $(x^i,u)$. We can directly employ the results given by \eqref{generators1}    in order to find  the components of the vector field \eqref{prolonged} for the family of wave equations  by considering the relations \eqref{notation}.
We shall write the following using the very well-known chain rule, 
 \[
 \frac{\partial}{\partial v_3}\rightarrow \frac{\partial}{\partial v_3}+ \frac{\partial}{\partial \Sigma^3} \frac{\partial \Sigma^3}{\partial v_3} =\frac{\partial}{\partial v_3}-\frac{\partial}{\partial \Sigma^3}.
\]
From this it is straightforward to notice that
 \be 
 \label{determining1}
 \zeta_3+\mu^3=0
 \ee
 where the relevant infinitesimal generators are explicitly
\begin{align*}
& \zeta_3= D_t \eta - (D_t \phi^j) v_j,\\
 &\mu^3= (D_t \phi^3-w-\frac{\partial \phi^j}{\partial u} v_j +\alpha^{33})v_3- (D_x \phi^3) f- (D_y \phi^3)g -\alpha^{13} v_1-\alpha ^{23} v_2+\beta^3
\end{align*}
 Using $\zeta_3$ and $\mu^3$ in the determining equation \eqref{determining1} we have the following results
 \be 
 \phi^3=\phi^3(t), \quad \alpha^{13} =\phi^1_t, \quad \alpha^{23}=\phi^2_t, \quad \beta^3 =- \eta_t , \quad w=\alpha^{33} +2 \dot \phi^3 +\eta_u.
 \ee
 Thus the infinitesimal generators of most general admissible equivalence transformation for the wave type equations given by \eqref{main} are explicitly written as:
 \be
\label{generators2}
\bs
& \xi^1=-\phi^1 (x,y,t,u), \ \xi^2=-\phi^2 (x,y,t,u), \ \xi^3=-\phi^3(t), \ \eta=\eta (x,y,t,u),\\
& \zeta_1= \eta_x+\eta_u v_1 +(\xi^1_x+\xi^1_u v_1)v_1+(\xi^2_x+\xi^2_u v_1)v_2,  \\
& \zeta_2=  \eta_y+\eta_u v_2 +(\xi^1_y+\xi^1_u v_2)v_1+(\xi^2_y+\xi^2_u v_2)v_2,  \\
& \zeta_3=  \eta_t+\eta_u v_3 +(\xi^1_t+\xi^1_u v_3)v_1+(\xi^2_t+\xi^2_u v_3)v_2+ \dot \xi^3 v_3  \\
&\mu^1=(\eta_u +\alpha^{33}+2 \dot \xi^3+\xi^2_u v_2-\xi ^1_x)f -(\xi^1_y+\xi^1_u v_2) g +\alpha^{11} v_1+\alpha^{12} v_2\\
&\qquad +(2\xi^1_t +\xi^1_u v_3) v_3 +\beta^1,\\
& \mu^2 = (\eta_u +\alpha^{33}+2 \dot \xi^3+\xi^1_u v_1-\xi ^2_y)g- (\xi^2_x+\xi^2_u v_1) f   -\alpha^{12} v_1+\alpha^{22} v_2 \\
&\qquad +(2\xi^2_t +\xi^2_u v_3) v_3 +\beta^2 
 \end{split}
\ee
and $\mathcal{H}$ can  be  evaluated by the last formula given in \eqref{generators1}. \\
These expressions  do not impose any restriction on functional dependencies of the smooth functions $f, \ g$ and $h$.  If some variables do not appear in the coordinate cover of
the extended manifold \eqref{manifold}, due to a particular structure of the given differential  equation, that might entail some restrictions on the infinitesimal generators of the group because the
corresponding  components on the extended manifold must then be set  zero. In the following section we shall investigate these in details.\\
\section{Admissible Transformations Between Linear and Nonlinear Equations}
Here we examine the possibilities of linearization problem for two dimensional nonlinear wave equations via point transformations. It is obvious that the infinitesimal generators \eqref{generators2} for the general wave type equations allow mappings between linear and nonlinear members as the forms of the transformations are
\be
\label{generaltransformations}
\bar x= \bar x(x,y,t,u,\epsilon), \quad \bar y=\bar y(x,y,t,u,\epsilon),\quad \bar t=\bar t(t,\epsilon), \quad \bar u=\bar u(x,y,t,u,\epsilon).
\ee
These  transformations span a wide group of mappings between nonlinear and linear members of the family \eqref{main}. Some examples are given in the next section. \\
But it is certain that not every  nonlinear wave equation be transformed into a linear equation. Here we shall try to give answer to the question  "which type of nonlinear equations can be linearized under the subsets of these point transformations?". Precisely, what are the functional dependencies of the arbitrary functions $f,\ g$ and $h$ in order to ensure the nonlinear equation is mapped onto a linear one or vice versa.\\
To answer such questions and construct the nonlinear structure of the arbitrary functions that can be linearized,  we  restrict their functional dependencies  and examine the effects on the infinitesimal generators, namely on the equivalence transformations.
Every particular restriction on the functional dependencies makes the corresponding additional variable in  \eqref{additional} vanish, so that corresponding component of the vector field \eqref{prolonged} must be taken zero. The components of the vector field for the additional variables \eqref{additional} are evaluated by the following formulas:
 \begin{equation}
  \label{additional_components}
  \begin{split}
&S^i_j=\frac{\partial F^i}{\partial x^j}+\frac{\partial
F^i}{\partial \Sigma^k}s^k_j+\frac{\partial F^i}{\partial
\Sigma}t_j,\ \
{\cse^i}=\frac{\partial F^i}{\partial u}+\frac{\partial F^i}{\partial \Sigma^k}\sigma^k+\frac{\partial F^i}{\partial \Sigma}\tau, \\
& S^{ij}=\frac{\partial F^i}{\partial v_j}+\frac{\partial
F^i}{\partial \Sigma^k}s^{kj}+\frac{\partial F^i}{\partial
\Sigma}t^j,  \ \
  T_i=\frac{\partial G}{\partial x^i}+\frac{\partial G}{\partial \Sigma^j}s^j_i+\frac{\partial G}{\partial \Sigma}t_i, \\
& \cT=\frac{\partial G}{\partial u}+\frac{\partial G}{\partial
\Sigma^j}\sigma^j+\frac{\partial G}{\partial \Sigma}\tau, \ \
T^i=\frac{\partial G}{\partial v_i}+\frac{\partial G}{\partial
\Sigma^j}s^{ji}+\frac{\partial G}{\partial \Sigma}t^i,
\end{split}
\end{equation}
 where $F^i=-s^i_j \xi^j-\sigma^i \eta-s^{ij}\zeta_j+\mu^i,\quad
G=-t_i \xi^i-\tau \eta -t^i \zeta_i+\mathcal{H}$.

 \subsection{Case  $\cfrac{\partial f}{\partial u_t}=0$}
 Let the arbitrary function $f$ be  $f=f(x,y,t,u,u_x,u_y)$ in the equation \eqref{main}.
 \begin{lemma}
 The particular form of  wave type equation 
 \be
 \label{fut}
 f(x,y,t,u,u_x,u_y)_x+  g(x,y,t,u,u_x,u_y,u_t)_y+  h(x,y,t,u,u_x,u_y,u_t) =u_{tt}
 \ee
admits the following equivalence transformations:
\be
\label{futtransformations}
\bar x= \bar x(x,\epsilon), \quad \bar y=\bar y(x,y,t,u,\epsilon),\quad \bar t=\bar t(t,\epsilon), \quad \bar u=\bar u(x,y,t,u,\epsilon).
\ee
\end{lemma}
\begin{proof}
When $\frac{\partial f}{\partial u_t}=0$,  one of the additional variables \eqref{additional} in the extended manifolds coordinate cover vanish: $s^{13}=0$, thus the corresponding component on the vector field must be taken identically zero: 
 \be
 \label{s13}
 S^{13}=0.
 \ee
 Writing that explicitly from \eqref{additional_components} we have the following
 \[
 S^{13}= 2 \frac{\partial F^1}{\partial v_3}+ \frac{\partial  F^1}{\partial g} s^{23} +\frac{\partial F^1}{\partial h} t^3, \quad F^1= -s^1_1 \xi^1-s^1_2 \xi^2 -s^1_3 \xi^3-\sigma^1 \eta- s^{11} \zeta_1 -s^{12} \zeta_2+ \mu^1.
 \]
Note that  $F^1$ does not depend neither  on $s^{23}$ nor on $t^3$,  thus  the  determining equation \eqref{s13} yields
\bes 
\frac{\partial F^1}{\partial v_3} = \frac{\partial F^1}{\partial g}=0.
\ees
Moreover $F^1$ involves $s^{11}$ and $s^{12}$ without any dependence of its components of vector field coefficients, moreover because $\xi^i$`s are independent of both $v_3$ and $g$,  these equations can be reduced into
\be\label{s13v3}
\frac{\partial \zeta_1}{\partial v_3} = \frac{\partial \zeta_2}{\partial v_3} = \frac{\partial \mu^1}{\partial v_3} =0
\ee 
and 
\be\label{s13g}
\frac{\partial \mu^1}{\partial g}=0.
\ee
From \eqref{generators2}, one can simply see that the first two equations of \eqref{s13v3} are directly satisfied since $\zeta_1 \text{ and } \zeta_2$ are independent from $v_3$, whereas the last equality, by using $\mu^1$ from \eqref{generators2} gives
\[
\frac{\partial \xi^1}{\partial t}=  \frac{\partial \xi^1}{\partial u}=0.
\] 
Equation \eqref{s13g} by a similar analysis  yields 
\[
\frac{\partial \xi^1}{\partial y}=0.
\]
Thus the last results ensure that the infinitesimal generator $\xi^1$ and the equivalence transformation for the local coordinate $x$ must be in the following form
\be\label{s13xi1}
\xi^1 = \xi^1 (x)\quad \Rightarrow \quad \bar x=\bar x(x,\epsilon).
\ee
Because there is no restriction on the other generators, the proof ends here.
\end{proof}
As a result of the Lemma, since the transformations of  both $\bar y$ and $\bar u$  involve the dependent variable $u$, we write the following corollary.
\begin{corollary}
The special type of wave equations given by \eqref{fut} admit maps between its linear and nonlinear members via the equivalence transformations given by \eqref{futtransformations}.
\end{corollary}
\subsubsection{Case  $\cfrac{\partial f}{\partial u_t}=h=0$}
In addition to the previous case when there is no nonhomogeneous term in the family of equations; $h=0$, transformations between linear and nonlinear members are  restricted by the following theorem. 
\begin{theorem}
 The particular family  wave type equations \[
 f(x,y,t,u,u_x,u_y)_x+  g(x,y,t,u,u_x,u_y,u_t)_y =u_{tt}
 \]
admits the  equivalence transformations between its linear and nonlinear members iff the transformation of the local coordinate $\bar y$ involves the dependent variable $u$.
\end{theorem}
\begin{proof}
When $h=0$ in the equation, its corresponding infinitesimal generator in the vector field \eqref{prolonged} must be set zero, calling that form \eqref{generators1}, we have explicitly
\be
\label{mathcalh}
\mathcal{H}=\frac{\partial \mu^1}{\partial x}+\frac{\partial \mu^1}{\partial u} v_1+ \frac{\partial \mu^2}{\partial y}+\frac{\partial \mu^2}{\partial u}v_2- \frac{\partial \zeta_3}{\partial t}-\frac{\partial \zeta_3}{\partial u}v_3=0.
\ee
Inserting $\mu^1, \mu^2$ and $\zeta_3$ into the equation, by remembering that the infinitesimal generators for local coordinates are found for the restriction $\frac{\partial f}{\partial u_t}=0$ in the previous case as
$
\xi^1=\xi^1(x),\ \xi^2=\xi^2(x,y,t,u), \ \xi^3=\xi^3(t),\ \eta=\eta(x,y,t,u)
$,
we have 
\begin{align}
&(\eta_u+\alpha^{33}-\xi^2_y)_u= (\eta_u+\alpha^{33}-\xi^2_y)_y=(\eta_u+\alpha^{33}-\xi^2_y-\xi^1_x)_x=  0,\label{ilk}\\
& (\eta_u -\xi^2_y)_u= (\eta_u -\xi^2_y+\frac{1}{2}\dot \xi^3)_t=0\label{iki}.
\end{align}
From \eqref{ilk} 
$$
\eta_u+\alpha^{33}-\xi^2_y= \frac{d \xi^1(x)}{dx }+a(t)
$$
can simply be obtained, where $a(t)$ is an arbitrary continuous function. Using this result in the second equality of \eqref{iki}, $\alpha^{33}$ is found as
$$\alpha^{33} =\frac{1}{2} \dot \xi^3 +a(t)+\lambda(x,y,u)$$ and substituting this into the previous equation we get
$$ 
\eta_u-\xi^2_y= \frac{d \xi^1(x)}{dx }-\frac{1}{2}\frac{d \xi^3(t)}{dt}-\lambda(x,y,u).
$$
The first part of the equation \eqref{iki} requires that $\frac{\partial \lambda(x,y,u) }{\partial u}=0$. Thus the infinitesimal generator related to the dependent variable is obtained as
\be
\label{eta1}
\eta= (\frac{d \xi^1(x)}{dx }-\frac{1}{2}\frac{d \xi^3(t)}{dt}-\lambda(x,y))u + \int \xi^2 (x,y,t,u) du.
\ee
Depending on the functional dependence of the infinitesimal generator $\xi^2$ there are two alternatives.
\begin{itemize}
\item if $\xi^2= \xi^2(x,y,t)$, then from \eqref{eta1} we have the infinitesimal generator  $$
\eta= (\frac{d \xi^1(x)}{dx }+\xi^2 (x,y,t)-\frac{1}{2}\frac{d \xi^3(t)}{dt}-\lambda(x,y))u+\gamma_1(x,y,t)$$
which does not let $\eta$ depends on $u$ nonlinearly. Thus  the equivalence transformations can be set as
\[
\bar x= \bar x(x,\epsilon),\quad \bar y=\bar y(x,y,t,\epsilon), \quad \bar t=\bar t(t,\epsilon),\quad \bar u=\alpha_1(x,y,t,\epsilon)u+\alpha_2(x,y,t,\epsilon)
\]
where $ \alpha_2$ is a continuously differentiable arbitrary function and $\alpha_1$ is given above. One can simply see that such type of transformations do not map nonlinear and linear  equations onto each other.
\item if $\xi^2= \xi^2(x,y,t,u)\text{ with }  \xi^2_u\neq 0$,  maps between linear and nonlinear equations can be derived via the following  transformations with $\bar y_u\neq 0$.
\[
\bar x= \bar x(x,\epsilon),\quad \bar y=\bar y(x,y,t,u,\epsilon), \quad \bar t=\bar t(t,\epsilon),\quad \bar u=\bar u (x,y,t,u,\epsilon)
\]
which completes the proof.
\end{itemize}
\end{proof}
\subsubsection{Case $\cfrac{\partial f}{\partial u_t}=\cfrac{\partial f}{\partial u_y}=0$}
If the particular  wave equation \eqref{main} is taken with such restriction then  the additional variables $s^{12},\ s^{13}$ are equal to zero, thus their corresponding infinitesimal generators are
$$S^{12}=S^{13}=0.$$
We have already examined the effects of $S^{13}=0$ in the subsection 3.1. The other restriction $S^{12}=0$ drops the dependence of dependent variable in the local transformations but as the following theorem states,  transformations between linear and nonlinear equations are still possible by taking $\eta_{uu}\neq 0$. 
\begin{theorem}
The nonlinear wave type equation 
$$
f(x,y,t,u,u_x)_x+g(x,y,t,u,u_x,u_y,u_t)_y+h(x,y,t,u,u_x,u_y,u_t)=u_{tt}
$$
can only be linearized under an appropriate point transformation that $\bar u=\bar u(x,y,t,u,\epsilon) \text{ with } \frac{\partial^2 \bar u}{\partial u^2}\neq 0$.
\end{theorem}
\begin{proof}
After doing similar computations have been done in the previous cases, restriction $S^{12}=0$  from \eqref{additional_components}  explicitly yields 
\[\frac{\partial F^1}{\partial v_2}=0\quad  \Rightarrow\quad \frac{\partial \zeta^1}{\partial v_2}=0, \quad \frac{\partial \mu^1}{\partial v_2}=0
\]
which restricts $\xi^2$ to
$$\frac{\partial \xi^2}{\partial x}=\frac{\partial \xi^2}{\partial u}=0 \quad \Rightarrow\quad \xi^2=\xi^2(y,t).$$
Thus the relevant generators become
\be
\label{futuy}
\xi^1=\xi^1(x),\quad \xi^2=\xi^2(y,t), \quad \xi^3=\xi^3(t),\quad \eta=\eta(x,y,t,u)
\ee
which points out that maps between linear and nonlinear members of the wave type equation given in the theorem can only  be generated by  the nonlinear dependence of $\eta$ on $u$. 
\end{proof}
\subsubsection{Case $\cfrac{\partial f}{\partial u_t}=\cfrac{\partial f}{\partial u_y}=h=0$}
In addition to the previous case when $h=0$, the structure of the equivalence transformations changes dramatically. The following theorem explains the result.
\begin{theorem}
For the wave type equations given by
$$
f(x,y,t,u,u_x)_x+g(x,y,t,u,u_x,u_y,u_t)_y=u_{tt}
$$
there is no admissible point transformation between its linear and nonlinear members.
\end{theorem}
\begin{proof}
Infinitesimal generators for $f_{u_t}=f_{u_y}=0$ are given by  \eqref{futuy}  in the previous case. Determining equation for vanishing $h$ is also given in \eqref{mathcalh} explicitly. To examine the effect of this restriction let us write $\mu^1,\ \mu^2$ and $\zeta_3$ under the influence of \eqref{futuy}
 \beas
 &&\mu^1= (\alpha^{33}+ 2 \dot \xi^3 +\eta_u-\xi^1_x)f+\alpha^{11} v_1 +\beta^1,\\
 &&\mu^2= (\alpha^{33}+ 2 \dot \xi^3 +\eta_u-\xi^2_y)g + 2\xi^2 _t v_3 +\alpha^{22} v_2 +\beta^2,\\
 && \zeta_3 =\eta_t + \xi^2 _t v_2 +(\eta_u+\dot \xi^3)v_3.
 \eeas
 Substituting these into the equation \eqref{mathcalh} gives the following relations
 \bea
&& (\alpha^{33}+\eta_u)_u=0, \quad (\alpha^{33}+\eta_u-\xi^1_x)_x=0, \quad (\alpha^{33}+\eta_u-\xi^2_y)_y=0, \\
&& \eta_{uu}=0,\quad (\xi^2- \eta_u)_t=\frac{1}{2}\dot \xi^3.\label{eta}
\eea
The first part of equations \eqref{eta} declares that the nonlinear dependence on the dependent variable $u$ is no more admissible for $\bar u$.  Therefore there is no transformation between  linear and nonlinear members of the given family of equations. The explicit solutions of the infinitesimal generators are determined as follows
\[
\xi^1= \xi^1(x),\quad \xi^2=\xi^2(y,t),\quad \xi^3=\xi^3(t), \quad \eta= (\xi^2_y-\frac{1}{2}\dot \xi^3+\lambda(x,y))u+\gamma(x,y,t).
\]
\end{proof}

 \subsection{Case $\cfrac{\partial f}{\partial u_t}=\cfrac{\partial g}{\partial u_t}=0$}
Examination of the absence of $u_t$ in the arbitrary function $f$ is studied in the case 3.1, in addition to this  its absence in $g$  can similarly be done.   If both $f$ and $g$ do not depend on  $u_t$,  the following Lemma states the form of infinitesimal generators  by the solution of determining equations $S^{13}=S^{23}=0$ related to the  vanishing additional variables $s^{13}=s^{23}=0$.
\begin{lemma}
The particular family of wave type equations 
$$
f(x,y,t,u,u_x,u_y)_x+g(x,y,t,u,u_x,u_y)_y+h(x,y,t,u,u_x,u_y,u_t)=u_{tt}
$$
admits transformations between its linear and nonlinear members via the infinitesimal generators
\be
\label{futgut}
\xi^1=\xi^1(x,y),  \quad \xi^2=\xi^2(x,y),   \quad \xi^3=\xi^3(t),\quad \eta=\eta(x,y,t,u), \quad \eta_{uu}\neq 0.
\ee
\end{lemma}
\begin{proof}
The proof  is straightforward. 
\end{proof}
In addition to the generators \eqref{futgut}, the remaining infinitesimal generators are 
\be
\label{uc}
\begin{split}
&
 \zeta_1= \eta_x+ (\eta_u + \xi^1_x ) v_1 +\xi^2_ x v_2, \quad \zeta_2= \eta_y + \xi^1_y v_1+(\eta_u+\xi^2_y) v_2,\\
 & \zeta_3= \eta_t+ (\eta_u+\dot \xi^3)v_3,\\
 &  \mu^1= (\alpha^{33}+ 2 \dot \xi^3+\eta_u-\xi^1_x)f- \xi^1_y g+ \alpha^{11} v_1+\alpha^{12} v_2+\beta^1,\\
 & \mu^2= (\alpha^{33}+ 2 \dot \xi^3+\eta_u-\xi^2_y)g- \xi^2_x f- \alpha^{12} v_1+\alpha^{22} v_2+\beta^2,\\
 \end{split}
\ee
and $\mathcal{H}$ can be evaluated by the formula given in \eqref{mathcalh}. 
\subsubsection{Case $\cfrac{\partial f}{\partial u_t}=\cfrac{\partial g}{\partial u_t}=h=0$}
For this case the admissible point transformations changes dramatically as given in the following theorem.
\begin{theorem}
The family of wave type equations given by
\[
f(x,y,t,u,u_x,u_y)_x+g(x,y,t,u,u_x,u_y)_y=u_{tt}
\]
does not admit any point transformation that maps linear and nonlinear equations into each other. 
\end{theorem}
\begin{proof}
Infinitesimal generators for  nonexistence of the term $u_t$ on both $f$ and $g$  are given in the previous case by \eqref{futgut}. Using  $\mathcal{H}=0$ explicitly given by \eqref{mathcalh},   one can determine that $\eta$ has to be
$$\eta= \xi^1_x-\xi^2_y-\frac{1}{2}\dot \xi^3-\lambda(x,y,))u+\gamma(x,y,t).
$$
Since neither the infinitesimal generators of the local coordinates depend on $u$ nor $\eta$ involves $u$ nonlinearly, there is no point transformation that would transfer linear and nonlinear members into each other for such type of family of wave equations. 
\end{proof}
\begin{table}[]
\caption{Infinitesimal generators related to  the local coordinates $x,y$ and $u$ for particular functional dependencies of the family of equations $f_x+g_y+h=u_{tt}$. 
}
\label{table:infinitesimal_table}
\footnotesize
\begin{tabular}{llll}
\hline
Section   & $ f,\ g,\ h$   & $\xi^1,\ \xi^2,\  \eta, \ \xi^3=\xi^3(t),$             & admissible mapping 
\\
\hline
\hline
3 & \begin{tabular}[c]{@{}l@{}}$f(x,y,t,u,u_x,u_y,u_t)$ \\ $g(x,y,t,u,u_x,u_y,u_t)$ \\ $h(x,y,t,u,u_x,u_y,u_t)$ \end{tabular}

 & \begin{tabular}[c]{@{}l@{}}$\xi^1(x,y,t,u)$\\ $\xi^2(x,y,t,u)$\\ $\eta(x,y,t,u)$\end{tabular}
 
  & \begin{tabular}[c]{@{}l@{}} L $\Longleftrightarrow$ N \\  
   w. $\xi^1_u \neq 0 /  \xi^2_u \neq 0$ \\ / $\eta_{uu}\neq 0$ \end{tabular} \\
\hline
3.1   & \begin{tabular}[c]{@{}l@{}}$f(x,y,t,u,u_x,u_y)$ \\ $g(x,y,t,u,u_x,u_y,u_t)$ \\ $h(x,y,t,u,u_x,u_y,u_t)$ \end{tabular} 
    & \begin{tabular}[c]{@{}l@{}}$\xi^1(x)$\\ $\xi^2(x,y,t,u)$\\ $\eta(x,y,t,u)$\end{tabular} 
    &  \begin{tabular}[c]{@{}l@{}} L $\Longleftrightarrow$ N   \\
    w. $\xi^2_u \neq 0$ /  $\eta_{uu}\neq 0$ \end{tabular} \\
     \hline
    -   & \begin{tabular}[c]{@{}l@{}}$f(x,y,t,u,u_x,u_y,u_t)$ \\ $g(x,y,t,u,u_x,u_y)$ \\ $h(x,y,t,u,u_x,u_y,u_t)$ \end{tabular} 
    & \begin{tabular}[c]{@{}l@{}}$\xi^1(x,y,t,u)$\\ $\xi^2(y)$\\ $\eta(x,y,t,u)$\end{tabular} & \begin{tabular}[c]{@{}l@{}} L $\Longleftrightarrow$ N   \\
    w. $\xi^1_u \neq 0$ /  $\eta_{uu}\neq 0$ \end{tabular}  \\ \hline
    3.2   & \begin{tabular}[c]{@{}l@{}}$f(x,y,t,u,u_x,u_y)$ \\ $g(x,y,t,u,u_x,u_y)$ \\ $h(x,y,t,u,u_x,u_y,u_t)$ \end{tabular} 
    & \begin{tabular}[c]{@{}l@{}}$\xi^1(x,y)$\\ $\xi^2(x,y)$\\ $\eta(x,y,t,u)$\end{tabular} & \begin{tabular}[c]{@{}l@{}} L $\Longleftrightarrow$ N   \\
    w.  $\eta_{uu}\neq 0$ \end{tabular}   \\ \hline
    3.1.1   & \begin{tabular}[c]{@{}l@{}}$f(x,y,t,u,u_x,u_y)$ \\
       $g(x,y,t,u,u_x,u_y,u_t)$ \\ $0$ \end{tabular} 
    & \begin{tabular}[c]{@{}l@{}}$\xi^1(x)$\\ \hdashline
 \begin{tabular}[c]{@{}l@{}} $\xi^2(x,y,t)$\\ $\eta =\gamma_1(x,y,t)u+\gamma_2(x,y,t)$\end{tabular}\\  \hdashline   
    $\xi^2(x,y,t,u)$\\ $\eta=\int \frac{\partial \xi^2}{\partial y} du +( \frac{d \xi^1}{dx}-\frac{1}{2} \dot \xi^3-\lambda(x,y))u+\gamma(x,y,t)$\end{tabular} &  \begin{tabular}[c]{@{}l@{}}  L/N $\Longleftrightarrow$ L/N  \\ \\ \begin{tabular}[c]{@{}l@{}} L $\Longleftrightarrow$ N   \\
    w. $\xi^2_u \neq 0$ \end{tabular} N \end{tabular}\\ \hline
    -   & \begin{tabular}[c]{@{}l@{}}$f(x,y,t,u,u_x,u_y,u_t)$ \\
       $g(x,y,t,u,u_x,u_y)$ \\ $0$ \end{tabular} 
    & \begin{tabular}[c]{@{}l@{}}$\xi^2(y)$\\ \hdashline
 \begin{tabular}[c]{@{}l@{}} $\xi^1(x,y,t)$\\ $\eta =\gamma_1(x,y,t)u+\gamma_2(x,y,t)$\end{tabular}\\  \hdashline   
    $\xi^1(x,y,t,u)$\\ $\eta=\int \frac{\partial \xi^1}{\partial x} du +( \frac{d \xi^2}{dy}-\frac{1}{2} \dot \xi^3-\lambda(x,y))u+\gamma(x,y,t)$\end{tabular} &  \begin{tabular}[c]{@{}l@{}}  L/N $\Longleftrightarrow$ L/N  \\ \\ \begin{tabular}[c]{@{}l@{}} L $\Longleftrightarrow$ N   \\
    w. $\xi^1_u \neq 0$  \end{tabular}\end{tabular}\\ \hline
  3.2.1   & \begin{tabular}[c]{@{}l@{}}$f(x,y,t,u,u_x,u_y)$ \\ $g(x,y,t,u,u_x,u_y)$ \\ $0$ \end{tabular} 
    & \begin{tabular}[c]{@{}l@{}}$\xi^1(x,y)$\\ $\xi^2(x,y)$\\ $\eta=(\xi^1_x+\xi^2_y-\frac{1}{2}\dot\xi^3+\lambda(x,y))u+\gamma(x,y,t)$\end{tabular} & L/N $\Longleftrightarrow$ L/N   \\ \hline 
     3.1.2   & \begin{tabular}[c]{@{}l@{}}$f(x,y,t,u,u_x)$ \\ $g(x,y,t,u,u_x,u_y,u_t)$ \\ $h(x,y,t,u,u_x,u_y,u_t)$ \end{tabular} 
    & \begin{tabular}[c]{@{}l@{}}$\xi^1(x)$\\ $\xi^2(y,t)$\\ $\eta(x,y,t,u)$\end{tabular} & \begin{tabular}[c]{@{}l@{}} L $\Longleftrightarrow$ N   \\
    w.   $\eta_{uu}\neq 0$ \end{tabular} \\ \hline 
     3.1.3   & \begin{tabular}[c]{@{}l@{}}$f(x,y,t,u,u_x)$ \\ $g(x,y,t,u,u_x,u_y,u_t)$ \\ $0$ \end{tabular} 
    & \begin{tabular}[c]{@{}l@{}}$\xi^1(x)$\\ $\xi^2(y,t)$\\ $\eta=\xi^2_y-\frac{1}{2}\dot\xi^3+\lambda(x,y))u+\gamma(x,y,t)$\end{tabular} & L/N $\Longleftrightarrow$ L/N   \\ \hline 
    -   & \begin{tabular}[c]{@{}l@{}}$f(x,y,t,u,u_x,u_t)$ \\ $g(x,y,t,u,u_y)$ \\ $h(x,y,t,u,u_x,u_y,u_t)$ \end{tabular} 
    & \begin{tabular}[c]{@{}l@{}}$\xi^1(x,t)$\\ $\xi^2(y)$\\ $\eta(x,y,t,u)$\end{tabular} & \begin{tabular}[c]{@{}l@{}} L $\Longleftrightarrow$ N   \\
    w.  $\eta_{uu}\neq 0$ \end{tabular}   \\ \hline 
     -   & \begin{tabular}[c]{@{}l@{}}$f(x,y,t,u,u_x)$ \\ $g(x,y,t,u,u_y)$ \\ $0$ \end{tabular} 
    & \begin{tabular}[c]{@{}l@{}}$\xi^1(x,t)$\\ $\xi^2(y)$\\ $\eta=(\xi^1_x-\frac{1}{2}\dot\xi^3+\lambda(x,y))u+\gamma(x,y,t)$\end{tabular} & L/N $\Longleftrightarrow$ L/N   \\ \hline 
     3.2.2   & \begin{tabular}[c]{@{}l@{}}$f(x,y,t,u,u_x,u_y)$ \\ $g(x,y,t,u,u_x,u_y)$ \\ $h(x,y,t,u,u_x,u_y)$ \end{tabular} 
    & \begin{tabular}[c]{@{}l@{}}$\xi^1(x,t)$\\ $\xi^2(x,y)$\\ $\eta=(\lambda(x,y)-\frac{1}{2}\dot\xi^3)u+\gamma(x,y,t)$\end{tabular} & L/N $\Longleftrightarrow$ L/N   \\ \hline
     
\end{tabular}
\end{table}

\subsubsection{Case $\cfrac{\partial f}{\partial u_t}=\cfrac{\partial g}{\partial u_t}=\cfrac{\partial h}{\partial u_t}=0$}
The results of the previous cases  hint at  that the dependence on $u_t$ is important in order to be linearized. To clarify that, here we shall investigate the family of nonlinear equations not depending on $u_t$ or its mixed derivatives.
\begin{lemma}
 (2+1) dimensional family of wave equations that can expressed as 
$$
f(x,y,t,u,u_x,u_y)_x+g(x,y,t,u,u_x,u_y)_y+h(x,y,t,u,u_x,u_y)=u_{tt}
$$
does not admit any point transformation that can map linear and nonlinear members into each other.
\end{lemma}
\begin{proof}
In addition to the  determining equations examined in  section 3.2, here the infinitesimal generator $T^3$ must vanish as $t^3=\frac{\partial h}{\partial u_t}=0$.    From  \eqref{additional_components} we have
\bes
\label{g}
\frac{\partial G}{\partial v_3}=0.
\ees
By resorting to $G$, the foregoing equation  is reduced to
$
\cfrac{\partial \mathcal{H}}{\partial v_3}=0
$. Since $\mu^1$ and $\mu^2$, which are the components of $\mathcal{H}$   do not involve $v_3$ (see \eqref{uc}), the following equation is obtained on the other component $\zeta_3$:
\[
\frac{\partial^2 \zeta_3}{\partial t\partial v_3}+ \frac{\partial^2 \zeta_3}{\partial u\partial v_3}v_3+\frac{\partial \zeta_3}{\partial u}=0.
\] 
Introducing $\zeta_3$ given in \eqref{uc} into the equation we observe  
\be 
\label{etason}
\eta=(\lambda(x,y)-\frac{1}{2}\dot \xi^3)u+\gamma(x,y,t)
\ee
where $\lambda$ and $\gamma$ are continuously differentiable functions of their arguments. This result with the results of the infinitesimal generators of local coordinates $\xi^1=\xi^1(x,y),  \quad \xi^2=\xi^2(x,y),   \quad \xi^3=\xi^3(t)$ from  \eqref{futgut} together ends the proof.
\end{proof}

 \begin{corollary}
In view of the results of Cases 3.1.3, 3.2.1 and 3.2.2, in order to linearize a nonlinear member of a family of (2+1) dimensional wave type equation, at least one of the arbitrary functions $f,\ g$ or $h$  must involve $u_t$.
 \end{corollary}
One can examine  existence of point transformations between linear and nonlinear equations for every other different choice of  functional forms of  arbitrary functions $f,\ g$ and $h$, by running similar procedures. We have given the  results examined in details in the manuscript  and some more  in the Table \ref{table:infinitesimal_table}.  The first column of the table shows in which section the problem is examined in the paper. $-$  means  calculations of that particular problem have not been given explicitly in the text. Second column express the infinitesimal generators  of the local coordinates whereas the last column indicate whether  transformations between linear and nonlinear equations are possible or not. L $ \Longleftrightarrow$ N states   the existence of linearization. If transformations map only a linear equation into another linear one or a nonlinear one into another nonlinear we represent that by L/N $ \Longleftrightarrow$ L/N.

 \section{Applications}
 In this section we shall give some applications  that map linear and nonlinear equations into each other. Infinitesimal generators for the general family, in their most general form is given by \eqref{generators2}. By taking  some continuously differentiable functions-as infinitesimal generators,  a class of nonlinear equations can be mapped onto linear equations. 
\subsection{$\xi^1= m(u),\ \xi^2=\xi^3=\eta=0$}
Let us take $\xi^1= m(u),\ \xi^2=\xi^3=\eta=0$ where $m(u)\in C^2$. 
Integrating the system of ordinary differential equations \eqref{system} under the initial conditions \eqref{initial} we reach to  the particular  equivalence transformations
\begin{align}
\label{m(u)}
& \bar x= x-\epsilon m(u),\ \bar y=y, \ \bar t= t,\nonumber\\
& \bar v_1= \frac{v_1}{1-\epsilon m'(u) v_1}, \ \bar v_2= \frac{v_2}{1-\epsilon m'(u) v_1}, \ \bar v_3= \frac{v_3}{1-\epsilon m'(u) v_1},\\
& \bar f= f-\frac{\epsilon m'(u) (g\ v_2-v_3^2)}{1-\epsilon m'(u) v_1}, \ \bar g= \frac{g}{1-\epsilon m'(u) v_1}, \, \bar h= \frac{h}{1-\epsilon m'(u) v_1}.\nonumber
\end{align}
By using the partial derivatives that can be written as:
\begin{align*}
\frac{\partial}{\partial \bar x} & = \frac{\partial}{\partial x} \frac{\partial x}{\partial \bar x} + \frac{\partial}{\partial x} \frac{\partial x}{\partial \bar u} \frac{\partial \bar u}{\partial \bar x} +\frac{\partial}{\partial u} \frac{\partial u}{\partial \bar u} \frac{\partial \bar u}{\partial \bar x} = (1+m'(u) \bar v_1) \frac{\partial}{\partial x} +\bar v_1 \frac{\partial}{\partial u} ,\\
\frac{\partial}{\partial \bar y} & = \frac{\partial}{\partial y} + \epsilon m'(u) \bar v_2 \frac{\partial}{\partial \bar x} +\bar v_2 \frac{\partial}{\partial u}  ,\\
\frac{\partial}{\partial \bar t} & = \frac{\partial}{\partial t} + \epsilon m'(u) \bar v_3 \frac{\partial}{\partial \bar x} +\bar v_3 \frac{\partial}{\partial u} .
\end{align*} 
 the transformed functions obtained in \eqref{m(u)} keep the family of equations invariant:
\be 
\bar f_{\bar x}+\bar g_{\bar y}+h=\bar v_{3_{\bar t}} \Rightarrow f_x+g_y+h=v_{3_{t}}.
\ee
and the transformed  functions $\bar f,\ \bar g, \ \bar h$ and $\bar v_3$ are written in terms of the transformed variables as:
\begin{align*}
& \bar f= \tilde f+ \epsilon m' (\bar u) \left( \frac{\bar u_{\bar t}^2}{1+\epsilon  m' (\bar u) \bar u_{\bar x} } -\bar u_{\bar y} \tilde g \right),\\
& \bar g= (1+ \epsilon  m' (\bar u) \bar u_{\bar x}) \tilde{g}, \quad \bar h= (1+ \epsilon  m' (\bar u) \bar u_{\bar x}) \tilde{h} , \quad \bar v_3= \bar u_{\bar t}
\end{align*}
where $\tilde{( )}$ represents the transformed form of the functions.\\
{\bf{Example 1:}}
By choosing $f=u_x,\ g=0$ and $h=0$, the set of the  point transformations obtained above is mapped  the nonlinear equation 
\be
\label{example}
 \left(\frac{\epsilon m'(\bar u) \bar u_{\bar t}^2+\bar u_{\bar x}}{1+\epsilon m'(\bar u) \bar u_{\bar x}}\right)_{\bar x}=  \bar u_{\bar t \bar t} 
\ee
onto the  classical constant coefficient wave equation in the plane $ 
u_{xx}=u_{tt}$.  
The  general solution of the wave equation  $u=\psi(y)(t-x)+ \phi(y) (t+x)$ generates a solution to the nonlinear equation  \eqref{example}:
$$
\bar u- \psi(\bar y)(\bar t-\bar x-\epsilon m(\bar u))-\phi(\bar y)(\bar t+\bar x+\epsilon m(\bar u))=0
$$ 
where $\psi(\bar y)$ and $\phi(\bar y)$ are arbitrary functions.
\subsection{$\xi^1=m(u),\ \xi^2=p(u)$}
More general transformations, by taking the infinitesimal generators of the local coordinates   $\xi^1=m(u),\ \xi^2=p(u),\ \xi^3=\eta=0$ are determined as: 
\begin{align*}
& \bar x= x-\epsilon m(u),\ \bar y=y-\epsilon p(u), \ \bar t= t,\ \bar u=u,\nonumber\\
& \bar v_1= \frac{v_1}{1-\epsilon (m' v_1+p' v_2)}, \ \bar v_2= \frac{v_2}{1-\epsilon (m' v_1+p' v_2)}, \ \bar v_3=  \frac{v_3}{1-\epsilon (m' v_1+p' v_2)},\\
& \bar f= \frac{(1-\epsilon m' v_1) f -\epsilon m' (  v_2 g- v_3^2) } {1-\epsilon (m' v_1+p' v_2)},\   \bar g= \frac{(1-\epsilon p' v_2)  g-\epsilon p'(v_1 f-v_3^2)}{1-\epsilon (m' v_1+p' v_2)}, \\
& \bar h= \frac{h}{1-\epsilon (m' v_1+p' v_2)}.\nonumber
\end{align*} 
The transformed functions become
\begin{align*}
&\bar f= (1+ \epsilon p' \bar u_{\bar y}) \tilde f -\epsilon m' \bar u_{\bar y} \tilde g + \frac{\epsilon m' \bar u_{\bar t}^2}{1+\epsilon(m' \bar u_{\bar x}+p' \bar u_{\bar y})},\\
&\bar g= (1+\epsilon m' \bar u_{\bar x}) \tilde{g} -\epsilon p' \bar u_{\bar x} \tilde f + \frac{\epsilon p' \bar u_{\bar t}^2}{1+\epsilon (m' \bar u_{\bar x}+p' \bar u_{\bar y})},\\
& \bar h= (1+\epsilon (m' \bar u_{\bar x} +p' \bar u_{\bar y}) \tilde h, \qquad \bar v_3=\bar u_{\bar t}.
\end{align*}
Note that in consequence of nonconstant  $m(u),$ and  $ p(u)$,  $\bar f$ and $\bar g$ involve $\bar u_{\bar x}, \ \bar u_{\bar y},\ \bar u_{\bar t}$ whatever $f,\ g$  are.  On the other hand  $\bar h$  depends on $\bar u_{\bar x}, \ \bar u_{\bar y}$ as long as $h$ is not  identically zero.  
 \subsection{$\xi^1= m(u,y),\ \xi^2=\xi^3=\eta=0$}
 Let $\xi^1= m(u,y),\ \xi^2=\xi^3=\eta=0$, then the equivalence transformations are determined as:
\begin{align*}
& \bar x= x-\epsilon m(u,y),\quad \bar y=y, \quad \bar t= t,\nonumber\\
& \bar v_1= \frac{v_1}{1-\epsilon m_u(u,y) v_1}, \ \bar v_2= \frac{v_2}{1-\epsilon m_u(u,y) v_1}, \ \bar v_3= \frac{v_3}{1-\epsilon m_u(u,y) v_1},\\
& \bar f= f-\frac{\epsilon [( m_y(u,y) +m_u(u,y) v_2) g+m_u(u,y) v_3^2]}{1-\epsilon m_u(u,y) v_1}, \ \bar g= \frac{g}{1-\epsilon m_u(u,y) v_1}, \\
& \bar h= \frac{h}{1-\epsilon m_u(u,y) v_1}.\nonumber
\end{align*} 
And the transformed functions of $f,\ g,\ h$  become
\begin{align*}
&\bar f= \tilde f- \epsilon [ m_{\bar y}+m_{\bar u} (1-\epsilon m_{\bar u} \bar u_{\bar x})\bar u_{\bar y}] \tilde g- \frac{\epsilon m_{\bar u} \bar u_{\bar t}^2}{1+\epsilon m_{\bar u} \bar u_{\bar x}}, \\
& \bar g=\tilde g (1+\epsilon m_{\bar u} \bar u_{\bar x}) , \quad \bar h= \tilde h (1+\epsilon m_{\bar u} \bar u_{\bar x}). \end{align*}
Note that   $\bar f$ depend on $\bar u_{\bar x} \text{ and } \bar u_{\bar t}$ whatever $\tilde f,\ \tilde g,\ \tilde h$ are.
 \subsection{$\xi^2= m(u,x),\ \xi^2=\xi^3=\eta=0$}
 Similar to the previous case, equivalence transformations are found as  
\begin{align*}
& \bar x= x,\quad \bar y=y-\epsilon m(u,x), \quad \bar t= t,\quad \bar u=u,\\
& \bar v_1= \frac{v_1+\epsilon m_x v_2 }{1-\epsilon m_u v_2}, \quad \bar v_2= \frac{v_2}{1-\epsilon m_u v_2}, \quad \bar v_3= \frac{v_3}{1-\epsilon m_u v_2},\\
& \bar f= \frac{f}{1-\epsilon m_u v_2},\quad \bar g=
g-\frac{\epsilon [(m_x +m_u) v_1) f+ m_u v_3^2]}{1-\epsilon m_u v_2} , \ \bar h= \frac{h}{1-\epsilon m_u v_2}. 
\end{align*} 
Transformed functions $\bar f,\ \bar g$ and $\bar h$ are then written as 
\begin{align*}
& f= (1+\epsilon m_{\bar u} \bar u_{\bar y})\tilde f, \
 \bar g= \tilde g - \epsilon [ m_{\bar x}+m_{\bar u} (1-\epsilon m_{\bar u} \bar u_{\bar y})\bar u_{\bar x}] \tilde g - \frac{\epsilon m_{\bar u} \bar u_{\bar t}^2}{1+\epsilon m_{\bar u} \bar u_{\bar y}},\\  
 &\bar h= (1+\epsilon m_{\bar u} \bar u_{\bar y} ) \tilde h.
\end{align*} 
 One can simply see   $\bar g$ involves  $\bar u_{\bar x} \text{ and } \bar u_{\bar t}$ for $m_u\neq 0$ even $\tilde g=0$.  
\section{Conclusion and Remarks}
In the present paper we have investigated the admissible linear and nonlinear  members of the general family of two dimensional wave equations that can be mapped onto each other. It is clear that because the family is taken in its most general form, it has not been possible to consider every particular member. We have tried to give as many cases as we can put in the article to make sure the reader understand the procedure, so that many other particular cases can be examined by simply applying the relevant restrictions  to the general infinitesimal generators derived here.  \\
 We have found the necessary functional dependencies of nonlinear equations that can be linearized and their  admissible equivalence transformations. Meanwhile it has been  shown that the term $u_t$ is a must for the existence of maps between linear and nonlinear equations via point transformations.

  \end{document}